\documentclass[11pt]{article}
\usepackage{graphicx,epsfig}

\usepackage{fullpage,xspace}
\usepackage{amsmath,upgreek}
\usepackage{amsfonts}
\usepackage{amssymb}
\usepackage{epstopdf}

\newtheorem{theorem}{Theorem}

\newtheorem{claim}[theorem]{Claim}

\newtheorem{corollary}[theorem]{Corollary}

\newtheorem{lemma}[theorem]{Lemma}

\newcommand{\Xomit}[1]{ }
\newenvironment{proof}[1][Proof]{\textbf{#1.} }{\ \rule{0.5em}{0.5em}}

\begin{document}

\title{The weighted $2$-metric dimension of trees \\ in the
non-landmarks model}

\author{Ron Adar\thanks{Department of Computer Science, University of Haifa, Haifa,
Israel. \texttt{radar03@csweb.haifa.ac.il.}} \and Leah
Epstein\thanks{Department of Mathematics, University of Haifa,
Haifa, Israel. \texttt{lea@math.haifa.ac.il.}}}

\date{}

\maketitle

\vspace{-0.056cm}

\begin{abstract}
Let $T=(V,E)$ be a tree graph with non-negative weights defined on
the vertices. A vertex $\tau$ is called a separating vertex for
$u$ and $v$ if the distances of $\tau$ to $u$ and $v$ are not
equal. A set of vertices $L\subseteq V$ is a feasible solution for
the non-landmarks model (NL), if for every pair of distinct
vertices, $u,v\in V\setminus L$, there are at least two vertices
of $L$ separating them. Such a feasible solution is called a
\textit{landmark set}. We analyze the structure of landmark sets for trees and design a linear time algorithm for
finding a minimum cost landmark set for a given tree graph.
\end{abstract}

\section{Introduction}
In this paper we study a generalization of the metric dimension problem, which is a problem of finding
a landmark set or a resolving set of a graph
\cite{Slat75,HM1976,Babai,Chvatal,MT,KRR1996,CEJO00,SSH02,ST04,HSV11,ELW,HN12}.
A landmark set is a subset of vertices $L\subseteq V$ of an
undirected graph $G=(V,E)$,  such that for any $u,v\in V$ ($u\neq
v$), there exists $\tau\in L$ with $d(u,\tau)\neq d(v,\tau)$,
where $d(x,y)$ denotes the number of edges in a shortest path
between $x$ and $y$. In this case, $\tau$ is called a separating
vertex for $u$ and $v$. Alternatively, it is equivalent to require
for $L$ to have a separating vertex for every pair $u,v\in
V\setminus L$.

This problem was introduced by Harary and Melter \cite{HM1976} and by Slater \cite{Slat75}. The problem was studied in the combinatorics literature \cite{Babai,CH+07,CEJO00,CZ03}, with respect to complexity \cite{KRR1996,BE+06,HSV11,HN12,DPL11,ELW}, and with respect to design of polynomial time algorithms for certain graph classes (sometimes even for the
weighted version), and in particular for paths and trees \cite{Slat75,HM1976,KRR1996,CEJO00,SSH02,ELW}.

The $k$-metric dimension problem (for an integer $k \geq 2$) \cite{kmetric1,AEfirst,kmetric2,kmetric3,kmetric4,kmetric5} is the problem of finding a
subset $L$ that has at least $k$ separating vertices for every
pair of distinct vertices $u$ and $v$.
If this is required for
every $u,v \in V$, the model is called AP (all-pairs model) (introduced in \cite{kmetric1} and independently in \cite{AEfirst}), and
if this is required for every $u,v \in V \setminus L$, the model is
called NL (non-landmarks model), which was introduced in \cite{AEfirst}.
In all cases, a valid solution is
called a landmark set. In the weighted case, a non-negative
rational cost is given for each vertex by a function
$c:V\rightarrow \mathbb{Q}^{+}$. For a set $U\subseteq V$,
$c(U)={\sum_{v\in U}}w(v)$ is the total cost of vertices in $U$,
and the goal is to find a landmark set $L$ with the minimum value
$w(L)$. For a given graph, the minimum cardinality of any landmark set is called the $k$-metric dimension, while the minimum cost of any set is called the weighted $k$-metric dimension. Yero, Estrada-Moreno, and  Rodr{\'{\i}}guez{-}Vel{\'{a}}zquez \cite{kmetric3} showed that computing the $k$ metric dimension of an arbitrary graph is NP-hard.

In this paper we study the case of trees for $k=2$, in the non-landmarks model (NL). In this model, a landmark set always exists for any $k$ as $V$ is a landmark set.
Let $T=(V,E)$ be a tree graph. Let $n=\left\vert V\right\vert $ be
the number of $T$'s vertices. The case of a path graph was completely analyzed in \cite{AEfirst} (for all values of $k$ and the two models, see also \cite{kmetric4} for the all-pairs model), and therefore we will assume that $T$ has at least one vertex of degree greater than $2$. It was shown in \cite{AEfirst} that a minimal cost landmark set for a path graph and $k=2$ consists of two or three vertices. Note that (by definition) every solution for AP is a solution for NL. Any subset of three vertices of a path form a landmark set (for each of the models), the two endpoints of the path also for a landmark set for both models. However, a solution consisting of the first two vertices of the path or the last two vertices of the path are landmark sets for NL but not for AP. For the case $k=1$, trees were studied in \cite{CEJO00,HM1976,KRR1996,Slat75,ELW}.

Given a tree, a vertex of degree at least $3$ is called a core vertex or a {\it core}. A vertex of degree $2$ is called a path vertex, and a vertex of degree $1$ is called a leaf.
As paths were completely studied in \cite{AEfirst}, we will consider trees that have at least one core.
For a core $v$, we often consider the subtrees creating by removing $v$ from the tree, and call them the subtrees of its neighbors (one subtree for each neighbor). We sometimes consider the BFS tree that is created by rooting the tree at $v$.
A subtree of a neighbor of a core $v$ that is a path (without any cores) is called a leg (or a standard leg) of $v$. If a leg consists of a single vertex (that is, $v$ is connected to a leaf), we called it a short leg, and otherwise it is called a long leg. For a leg $\ell$ of $v$ which consists of $j \geq 1$ vertices, the vertex of position $i$, also denoted by $\ell^i$ (for $1 \leq i \leq j$), is the vertex of the leg $\ell$ whose distance from $v$ is $i$.
Khuller et al. \cite{KRR1996} showed  that a landmark set for $k=1$ can be created by selecting all leaves except for the leaf of one leg of each core. In \cite{ELW}, it was shown that a minimum cost landmark set is created by selecting one vertex of each leg of a core except for one leg, such that the selected vertices have total minimal cost.

Given a tree, we define a {\it small core} to be a core vertex that satisfies all the following conditions; The core has a degree of exactly $3$, it has at least two legs, one of which is a short leg. The second leg of a small core may be either short or long, and the third subtree of a neighbor of $v$ can be a (short or long) leg or it can contain one or more cores. Other cores are called regular cores. If the third subtree of a neighbor of $v$ contains a core, the closest core to $v$, $x$, is connected to $v$ by a path consisting of path vertices. For a small core $v$, if the closest core, $x$, is a small core too, then since $x$ also has two legs (in addition to the subtree  of a neighbor of $x$ that contains $v$), the tree has exactly two (small) cores and no regular cores (while in any other case, $x$ is a regular core). We will have two special cases. One special case is where the tree contains exactly two small cores (and no regular cores). The other special case is where the tree has exactly one small core (and no regular cores, that is, the tree consists of a core with three legs, at least one of which is short). The two special cases will be considered separately after some properties will be discussed, while all cases where the tree has at least one regular core will be treated together.

Next, we define a {\it modified leg}. For a core vertex $v$, a subtree of a neighbor of $v$ that has exactly one core, and this core is a small core, is called a modified leg (in this case $v$ is the closest regular core to at least one small core). That is, a modified leg of $v$ consists of a path to a small core, a small core, and the standard legs of the small core (where the small core has one short leg and one leg that is short or long).
A g-leg of a core $v$ is a subtree of a neighbor of $v$ that is either a standard leg (short or long) or a modified leg (the terms short leg and long leg will refer only to standard legs). A position on a modified leg $\ell$ is defined exactly as it is defined for a standard leg, but if the small core (which is part of $\ell$) is in position $i\geq 1$, then there are two vertices with position $i+1$ on $\ell$, denoted by $\ell^a$ and $\ell^b$, where $\ell^b$ is the unique vertex of a short leg of the small core (if both its standard legs are short, the choice of which vertex of a short leg is $\ell^a$ and which one is $\ell^b$ is done arbitrarily).
As explained above, except for the special case treated later, if a small core $u$ that belongs to a modified leg of $v$, then $v$ is a regular core.

A regular core is called minor if one of the two following conditions holds. The first condition is that it has at most one g-leg (as its degree is at least $3$, there are at least two other subtrees of its neighbors that are neither standard legs nor modified legs, and each such subtree contains a regular core).
The second condition is that its degree is at least $4$, it has no modified legs, it has exactly two standard legs, one of which is short (the other leg is either short or long) and the other (at least two) subtrees of its neighbors are not g-legs and contain regular cores.  Thus, for a minor core $v$, no matter which condition out of the two is satisfied, there are always at least two subtrees of the neighbors of $v$ which are not g-legs, and it either has at most one g-leg (which is standard or modified), or it has two standard legs, at least one of which is short. If a regular core is not minor, then we call it a {\it main core}.

In the AP (all pairs model), two separations in a landmark set $L$ are required for any pair of distinct vertices and not only for the vertices of $V\setminus L$. In this case, it is sufficient to analyze cores (without splitting them into small cores and regular cores) and standard legs. The condition for a set $L \subseteq V$ to be a landmark set is defined on the sets of legs of cores. For each core, if there is a leg that does not have any vertex in $L$, any other leg must have at least two vertices in $L$. In particular, if a core has a single leg, it will not have any vertices in $L$ in a minimum cost landmark set (see \cite{kmetric1,kmetric2,kmetric3}).

\begin{figure} [h!]
\hspace{1.4in}
\includegraphics[angle=0,width=0.5\textwidth]{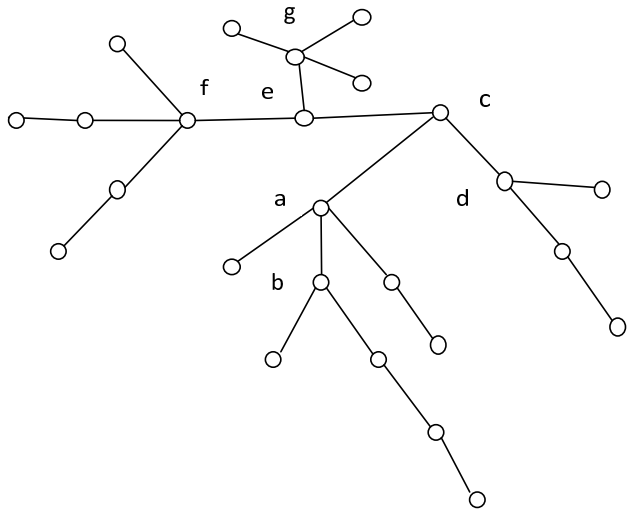}
\caption{An example of a tree with seven core vertices $a$,$b$,$c$,$d$,$e$,$f$,$g$. The vertices $b$ and $d$ are small cores, $c$ and $e$ are minor cores, all other are main cores. Note that $c$ is a core with a single (modified) leg, while $e$ is a core without any g-legs. \label{exmpl}}
\end{figure}

\section{Properties}
In this section we analyze the structure of landmark sets. We prove a number of lemmas and claims that determine the requirements of minimal landmark sets (with respect to cost or to set inclusion). We will define required solution types for g-legs. This will allow us to design a relatively simple algorithm for computing a minimum cost landmark set for a given tree in the next section.

\begin{lemma} \label{lemma1}
Consider a tree and a regular core $v$.
Every subtree of a neighbor of $v$ that has a regular core must have a main core. In
particular, a tree that has a minor core also has at least two
main cores, and every tree with a regular core has a main core.
\end{lemma}
\begin{proof}
Root the tree at $v$,and consider a subtree of $v$ that has a regular core. Let $x$ be a regular core of a largest distance to $v$  in this subtree (that is, a regular core of largest depth in the rooted tree). The core $x$ must have at least two children in the rooted tree, as its degree in the tree is at least $3$.  If the degree of $x$ in the tree is at least $4$, $x$ has at least three children in the rooted three, and the subtrees of these children are g-legs, as these subtrees contain no regular cores (since otherwise $x$ is not a regular core of maximum depth). Since $x$ has at least three g-legs, it is not a minor core, and therefore it is a main core. Consider the case where the degree of $x$ is $3$ and it has exactly two children in the rooted tree, where the subtrees of these children are g-legs. A minor core of degree $3$ only has one g-leg, and therefore this case is impossible. We find that $x$ is a main core.

Consider a minor core $y$. By definition, the subtrees of at least two neighbors of $y$ are not g-legs, each of them has a regular core and thus it also has main core, proving that there are at least two main cores in the tree.
Finally, consider a tree that has a regular core $u$. If $u$ is a main core, we are done. Otherwise, it is a minor core, and in this case there are at least two main cores in the tree.
\end{proof}


We define algorithms for finding subsets of vertices for g-legs of cores. We will call these sets {\it local sets}. We will show that for trees with at least one regular core, a minimum weight landmark set can be found by combining the local sets (in the two special cases of trees without regular cores, the conditions will still be necessary but they will not always be sufficient). In the algorithms, we will assume that a core with a set of g-legs is given. For a regular core this is simply the set of its g-legs. The case of small cores is relevant only for the two special cases. If the tree has a single core which is small, we consider the three standard legs of this core. In the case of a tree with two small cores, each one of them can be seen as a core with two standard legs (one of which is small), and one modified leg.

Each solution for the set of g-legs of one core $v$ will consist of finding a solution for each g-leg and combining such solutions, but the g-legs cannot be dealt with independently (and $v$ will never be selected as a part of a local set). We now define several types of solutions for standard legs and for modified legs. For standard legs, the solution types are as follows. A type $(s,0)$ solution is an empty set (that is, no vertices of the leg are selected). A type $(s,1)$ solution consists of one vertex of the leg whose position on the leg is at least $2$. A type $(s,2)$ solution consists of at least two vertices of the leg. Notice that solutions of types $(s,1)$ and $(s,2)$ valid only for long legs. A type $(s,3)$ solution consists of the vertex with position $1$ in the leg.
For a modified leg $\ell$, whose small core is at position $i$, the solution types are as follows. A type $(m,1)$ solution contains exactly one vertex out of $\ell^a$ and $\ell^b$, that is, the solution is either $\{\ell^a\}$ or it is $\{\ell^b\}$. A type $(m,2)$ solution does not have any of the vertices $\ell^a$ and $\ell^b$, it consists of at least two vertices with positions $i+2$ or larger, and possibly other vertices (whose positions are not $i+1$). A type $(m,3)$ solution contains at least two vertices, one of which is either $\ell^a$ or $\ell^b$ (it is possible that it contains both these vertices).

A subset $S$ of vertices of the g-legs of a core $v$, inducing solutions for the g-legs, is called a local set if it satisfies the following conditions.
\begin{enumerate}
\item There is at most one standard leg whose solution is of type $(s,0)$, and the other standard legs have solutions of types $(s,1)$, $(s,2)$, and $(s,3)$.
\item All modified legs have solutions of types $(m,1)$, $(m,2)$, and $(m,3)$.
\item If there is a standard leg whose solution is of type $(s,0)$, then no modified leg has a solution of type $(m,1)$.
\item If there is a long leg $\ell$ whose solution is of type $(s,0)$, then every long leg except for $\ell$ has a solution of type $(s,2)$.
\item If there is a short leg whose solution is of type $(s,0)$, then every long leg has a solution of type $(s,2)$ or $(s,3)$.
\end{enumerate}
As mentioned above, a short leg cannot have a solution of type $(s,1)$ or $(s,2)$. Thus, if some standard leg $\ell$ has a solution of type $(s,0)$, then all short legs (expect for $\ell$, if it is short) have solutions of type $(s,3)$.

\begin{lemma}\label{mustbelocal}
Consider a landmark set $L$, and a core $v$. The set $L$ contains a local set $S$ of $v$ as a subset.
\end{lemma}
\begin{proof}
First, we show that if the vertices of a g-leg that are in $L$ do not form any of the types of solutions defined above, then $L$ is not a landmark set. For standard legs, these types cover all possible solutions, thus we consider modified legs.
Consider a modified leg $\ell$ of a regular core $v$, let $u$ be its small core, and let $i$ be the position of $u$ on $\ell$.  Any solution that contains at least one of $\ell^a$ and $\ell^b$ is either of type $(m,1)$ or of type $(m,3)$. If the solution does not contain any of the vertices $\ell^a$ and $\ell^b$, these two vertices cannot be separated by any vertex that is not on the legs of $u$, as their paths to such vertices traverse $u$ (or end at $u$), and their distances to $u$ are equal. To obtain two separations between $\ell^a$ and $\ell^b$ (if none of $\ell^a$ and $\ell^b$ is in $L$), $L$ must contain at least two vertices whose positions on $\ell$ are at least $i+2$ (these vertices are on the same standard leg of $u$ as $\ell^a$). Thus, in this case the solution is of type $(m,2)$. This also proves the second condition of local sets.

Consider two standard legs of $v$, $\ell_1$ and $\ell_2$. If none of the vertices $\ell_1^1$ and $\ell_2^1$ is in $L$, then they can only be separated by vertices of $\ell_1 \cup \ell_2$, as they have equal distances to $v$, and all paths from $\ell_1^1$ and $\ell_2^1$ to vertices not in $\ell_1 \cup \ell_2$ traverse $v$ (or end at $v$). Thus, the case that both legs have type $(s,0)$ solutions is impossible. This proves the first condition of local sets (and that expect for at most one short leg with an $(s,0)$ type solution, the solutions of every short leg are of type $(s,3)$).
Consider a standard leg without any landmarks, $\ell_1$, and a leg $\ell_2$ that is either long or modified. We will show that $\ell_2$ has at least two landmarks unless $\ell_1$ is short and $\ell_2$ is long, in which case it may have a type $(s,3)$ solution, and in all other cases the possible types of solutions are those having at least two vertices, that is, $(s,2)$, $(m,2)$, and $(m,3)$ type solutions. Assume that $\ell_2$ has exactly one landmark.
The landmark of $\ell_2$ must be in the first position of $\ell_2$, as otherwise none of $\ell_1^1$ and $\ell_2^1$ is in $L$,
and there is just one separation between them (because only vertices of $\ell_1 \cup \ell_2$ can separate them). If $\ell_2$ is a modified leg, then as the positions of $\ell_2^a$ and $\ell_2^b$ on $\ell_2$ are at least $2$, this is not a solution of type $(m,1)$, and since it has a single vertex, it is not a solution of type $(m,2)$ or $(m,3)$, and we reach a contradiction. We are left with the case that  $\ell_2$ is long, and its solution is of type $(s,3)$. However, if $\ell_1$ is long as well, since $\ell_1^2,\ell_2^2 \notin L$ (and $\ell_1$ has no vertices in $L$ while $\ell_2$ only has $\ell_2^1$ in $L$), these two vertices will only have one separation, showing that the solution is not valid. This proves the last three conditions.
%
%
%
%
%
%
%
%
%
%
%
\end{proof}

\begin{lemma} \label{lemma3}
Consider a set $L \subseteq V$. If for every core $v$, the subset of $L$ that is restricted to the vertices of the g-legs of $v$ is a local set, then every pair of vertices $x,y \notin L$ with equal positions on the g-legs of $v$ has at least two separations in $L$.
\end{lemma}
\begin{proof}
If $x$ and $y$ are on the same g-leg of $v$, then this must be a modified leg $\ell$ of $v$, and $\{x,y\}=\{\ell^a,\ell^b\}$. Since $x,y \notin L$, the solution for $\ell$ must be of type $(m,2)$. In this case there are two vertices of $L$ on the leg that contains $\ell^a$, the distance of each such vertex to $\ell^a$ is smaller than its distance to $\ell^b$, and there are at least two separations between $x$ and $y$.

Consider the case where $x$ and $y$ are on different g-legs of $v$. We claim that any vertex $z$ on the leg of $x$ is closer to $x$ than it is to $y$. Let $\tilde{\ell}$ be the leg of $x$. The path from $y$ to $z$ traverses $\tilde{\ell}^1$, thus $d(y,z)=d(y,v)+1+d(\tilde{\ell}^1,z)=d(x,v)+1+d(\tilde{\ell}^1,z)=d(x,\tilde{\ell}^1)+2+d(\tilde{\ell}^1,z)$. However, $d(x,z) \leq d(x,\tilde{\ell}^1)+d(\tilde{\ell}^1,z)$, proving $d(x,z)<d(y,z)$. Thus, if the two g-legs (of $x$ and $y$) have at least two vertices of $L$ in total, then there are at least two separations between $x$ and $y$. Any modified leg has at least one vertex in any local set, and if there is a standard leg with a $(s,0)$ type solution, then any modified leg has at least two vertices in any local set. We find that the only case where the two g-legs have at most one vertex of $L$ (together) is where both these g-legs are standard legs, one leg have a type $(s,0)$ solution, and the other leg has a type $(s,3)$ solution. We show such a solution is not possible. If at least one of the legs of $x$ and $y$ is short, then it only has a vertex in position $1$. Thus, the positions of $x$ and $y$ are equal to $1$. This last case is impossible as $x,y \notin L$ implies that none of these two legs has a type $(s,3)$ solution if the positions of $x$ and $y$ are $1$. We find that both these legs are long, but then the leg that does not have a type $(s,0)$ solution must have a type $(s,2)$ solution, and the two legs have at least two vertices of $L$ in total, so this case is impossible too.
\end{proof}

\begin{lemma}\label{twotwo}
A local set of a main core has at least two vertices.
\end{lemma}
\begin{proof}
Consider a main core $v$ where the local set of its g-legs has at most one vertex.
Any local set of $v$ contains at least one vertex on every g-leg except for possibly one g-leg, and it has at least one vertex of every modified leg. Thus, $v$ has at most two g-legs.
If $v$ has at most one g-leg, then it is a minor core. If it has exactly two g-legs, then one of them has no vertices in the local set, while the other one has one vertex in this set. The g-leg with one vertex in the local set cannot be a modified leg, as in this case (where there is a standard leg with a $(s,0)$ type solution), it would have at least two vertices in the local set, by property $3$ of local sets. Thus, $v$ has two standard legs. If one them is short, then $v$ is a minor core again. Otherwise, as a long leg has a type $(s,0)$ solution, the other long leg has a type $(s,2)$ solution, and the local set has at least two vertices, contradicting the assumption.
\end{proof}

\begin{claim}\label{clm1}
Consider a tree, a core $v$, and a set $L \subseteq V$ that contains a local set for $v$. If $x \neq y$ are vertices of the subtree consisting of $v$ and its g-legs such that $x,y \notin L$ and $d(x,v)<d(y,v)$, then any vertex of $L$ expect, possibly, for the vertices on the g-leg of $v$ that contains $y$ separates them.
\end{claim}
\begin{proof}
Consider a vertex $z \in L$ that is not on any g-leg of $v$ that contains at least one of $x$ and $y$ (note that $y\neq v$, as $d(y,v)>0$, so $y$ is on a g-leg of $v$). The vertex $z$ separates $x$ and $y$ as
their paths to $z$ traverse $v$ (a path can start at $v$ if $x=v$, or the paths can end at $v$ if $z=v$), while $x$ and $y$ have different distances to $v$. If $x$ is on a g-leg that is not the g-leg of $y$, consider a vertex $z'$ on this g-leg. Since $d(x,z') \leq d(x,v)+d(v,z')<d(y,v)+d(v,z')$ while $d(y,z')=d(y,v)+d(v,z')$, $z'$ also separates $x$ and $y$.
\end{proof}

\begin{lemma}\label{subt}
Consider a tree that has at least two regular cores, and a set $L \subseteq V$. If for every regular core $v$, the subset of $L$ that is restricted to the vertices of the g-legs of $v$ is a local set, then for every pair of vertices $x,y \notin L$ in the subtree consisting of $v$ and its g-legs, there are two separations between $x$ and $y$.
\end{lemma}
\begin{proof}
By Lemma \ref{lemma1}, a tree with a minor core also has at least two main
cores. Thus, the tree has at least two main cores in any case. If
$d(x,v)=d(y,v)$, then there are two separations between $x$ and
$y$ in the subset of $L$ restricted to the g-legs of $v$, as it is
a local set. Otherwise, the tree has at least one additional main
core $u$ except for $v$, such that the vertices of $L$ on the
g-legs of $u$ form a local set for $u$, and therefore there are at
least two such vertices, by Lemma \ref{twotwo}. By Claim
\ref{clm1}, $x$ and $y$ are separated by the vertices of $L$ that
are in the local set which is the subset of $L$ restricted to the
g-legs of $u$.
\end{proof}

\begin{lemma}\label{toocor}
Consider a tree with at least two regular cores, and a set $L
\subseteq V$. If for every regular core $v$, the subset of $L$
that is restricted to the vertices of the g-legs of $v$ is a local
set, then $L$ is a landmark set.
\end{lemma}
\begin{proof}
Consider two vertices $x\neq y$, such that $x,y\notin L$. If $x$
and $y$ are in the subtree consisting of a regular core and its
g-legs, then there are two separations between them due to Lemma
\ref{subt}. Next, assume that $x$ is a vertex of a subtree
consisting of a main core $v$ and its g-legs (that is, $x=v$ or
$x$ is on    a g-leg of $v$), while $y$ is not in this subtree. Let
$v'$ be the neighbor of $v$ on the tree path from $v$ to $y$. As
the subtree of $v'$ is not a g-leg of $v$, it has at least one
regular core and therefore it has a main core, by Lemma \ref{lemma1}. Let $z$ denote such
a core, let $z_1,z_2 \in L$ be vertices of the g-legs of $z$, and
let $v_1,v_2 \in L$ be vertices of the g-legs of $v$ (all of which
must exist by Lemma \ref{twotwo}).  Consider the case $d(y,v)\leq
d(x,v)$. As $d(y,z_i)\leq d(y,v')+d(v',z_i)=d(y,v)-1+d(v',z_i)$,
while $d(x,z_i)=d(x,v)+1+d(v',z_i)\geq d(y,v)+1+d(v',z_i)$, for $i \in \{$1$,$2$\}$,
$z_1$ and $z_2$ separate $x$ and $y$. In the case $d(y,v) > d(x,v)$, we find
$d(y,v_i)=d(y,v)+d(v,v_i)$, for $i \in \{$1$,$2$\}$, while $d(x,v_i) \leq d(x,v)+d(v,v_i) <
d(y,v)+d(v,v_i)$, so $v_1$ and $v_2$ separate $x$ and $y$. If $x$
is a minor core or on a g-leg of a minor core $a$, root the tree
at $a$ (we let $a=x$ if $x$ is a minor core). There are at least
two subtrees of $a$ with main cores (by Lemma \ref{lemma1}) and therefore, with at least
two vertices of $L$ (in each). Since we assume that $y$ is not $a$
or on a g-leg of $a$, it is in one of these subtrees. Let $a'$
be the neighbor of $a$ that is the root of this subtree. If
$d(y,a)\leq d(x,a)$, then for every vertex $u$ in the subtree of
$a'$, $d(x,u)=d(x,a)+1+d(a',u) > d(y,a)+d(a',u)$ and $d(y,u)\leq
d(y,a')+d(a',u)=d(y,a)-1+d(a',u)$, thus every such vertex $u$
separates $x$ and $y$, and there are two separations for this pair
of vertices. Otherwise, $d(y,a) > d(x,a)$, and for every vertex
$w$ in the subtree of a different neighbor of $a$ that has a main
core (and thus at least two vertices of $L$) in its subtree,
$d(x,w)=d(x,a)+d(a,w)$ and $d(y,w)=d(y,a)+d(a,w)$, giving two
separations for $x$ and $y$.

Finally, assume that none of $x$ and $y$ is a regular core or on a g-leg of a regular core (as all cases where one of the two vertices satisfies this condition were considered). Root the tree at $x$. There are at least two subtrees (since a leaf must be a part of a g-leg), and each one must have a regular core (otherwise $x$ is a part of a g-leg). For any vertex $b$ in the subtree that does not contain $y$, $d(y,b)=d(y,x)+d(x,b)>d(x,b)$, giving two separations between $x$ and $y$ again.
\end{proof}

\begin{lemma}\label{oneforeach}
Consider a tree with exactly one core $v$ (that has at least three g-legs), and a set $L \subseteq V$ that contains a local set for $v$. In all the following cases $L$ is a landmark set of the tree.

\begin{enumerate}
\item The set $L$ contains at least one vertex of each g-leg of $v$.
\item The core $v$ has at least four g-legs.
\item The core $v$ has at least two g-legs that are modified legs.
\item The set $L$ contains $v$.
\end{enumerate}

\end{lemma}
\begin{proof}
In each one of the cases we consider two vertices $x,y \notin L$.
If $d(x,v)=d(y,v)$, then $x$ and $y$ have the same position in their g-legs and there are two separations between $x$ and $y$ since $L$ is a local set, by Lemma \ref{lemma3}. Otherwise, assume without loss of generality that $d(x,v)<d(y,v)$ (and thus $y\neq v$).

In the first two cases, $v$ has at least two additional g-legs except for the g-leg of $y$ each having at least one vertex of $L$. In the first case this holds as there are two additional g-legs, and every g-leg has a vertex of $L$. In the second case, $L$ contains a local set and every g-leg, except for at most one, has a vertex of $L$. There are at least four g-legs in total, there are at least two g-legs except for the g-leg of $y$ and a g-leg with no vertices of $L$ (if such a standard leg exists). Thus, there are at least two separations between $x$ and $y$ as by Claim \ref{clm1}, any vertex of $L$ separates $x$ and $y$ possibly except for vertices on the g-leg of $y$.

In the third case, if every g-leg has a vertex of $L$, then the property follows from the first case. Otherwise, there is a standard leg without any vertex of $L$, and therefore every g-leg, except for this leg, has at least two vertices of $L$, by property $3$ of local sets. Since at least two g-legs of $v$ are modified legs, there is a modified leg that is not the g-leg of $y$, and its vertices that belong to $L$ separate $x$ and $y$.

In the fourth case, it is sufficient to consider a tree and a set $L$ that do not satisfy any of the conditions of the first three cases. Thus, $v$ has three g-legs. As $d(x,v)<d(y,v)$, $v$ separates $x$ and $y$. At least one of the g-legs that are not the g-leg of $y$ has at least one vertex in $L$, and this vertex separates $x$ and $y$ as well.
\end{proof}
\begin{corollary}\label{localsetsaregreat}
Consider a tree that has at least one regular core, and a set $L
\subseteq V$. If for every regular core $v$, the subset of $L$
that is restricted to the vertices of the g-legs of $v$ is a local
set, then $L$ is a landmark set.
\end{corollary}
\begin{proof}
If the tree has at least two regular cores, this claim was proved in Lemma \ref{toocor}. Assume that the tree has one regular core $u$. If $u$ has at least four g-legs, or it has three g-legs, out of which at least two g-legs are modified legs, then the claim was proved in Lemma \ref{oneforeach} (the second and third parts). Next, assume that $u$ has three g-legs, out of which at most one is modified. If $u$ has three g-legs, such that all of them are standard, and at least one leg is short, then $u$ is a small core (and not regular). If $u$ has one modified leg, and two standard legs, out of which at least one leg is short, then $u$ is a small core (in this case $u$ is not regular either, the tree has two small cores and no regular cores). Thus, we find that $u$ has two long legs, and the third g-leg is either modified, or it is standard and long.
If every leg of $u$ contains a vertex of $L$, then the claim was proved in Lemma \ref{oneforeach} (the first part). Otherwise, $u$ has a standard leg with no vertices of $L$. By the properties of local sets, there are two g-legs with at least two vertices of $L$ on each g-leg. In this case, for any pair $x,y\notin L$, if $d(x,u)=d(y,u)$, then there are two separations between $x$ and $y$ as $L$ contains a local set for the g-legs of $u$, and if $d(x,u)<d(y,u)$, then the two vertices of $L$ that are on a g-leg that does not contain $y$ separate $x$ and $y$.
\end{proof}

\begin{lemma}\label{threev}
Consider a tree with no regular cores and a single small core $v$, such that $v$ has three standard legs. If $L \subseteq V$ contains a local set for $v$ and $|L|\geq 3$, then $L$ is a landmark set. If $L \subseteq V$ contains a local set for $v$ and $|L| \leq 2$, then $L$ is a landmark set if and only if $L$ consists of the two vertices of two short legs of $v$.
\end{lemma}
\begin{proof}
Assume that $|L|\geq 3$. By Lemma \ref{oneforeach}, if $v \in L$, or if every leg has a vertex of $L$, we are done.
Otherwise, consider a pair $x,y \notin L$.
If $d(x,u)=d(y,u)$, then by Lemma \ref{lemma3} there are two separations between $x$ and $y$ as $L$ contains a local set for the g-legs of $u$.
If $d(x,u)<d(y,u)$, and the two legs that are not the leg of $y$ have at least two vertices of $L$, we are done too (by Claim \ref{clm1}).
As $|L|\geq 3$, the remaining case is that the leg of $y$ has at least two vertices in $L$, while one of the other legs of $u$ has one vertex of $L$, and it separates $x$ and $y$.
At most one vertex of the leg of $y$ can have equal distances to $x$ and $y$, and therefore, any other vertex of the leg of $y$ which is in $L$ separates them as well, and we assumed there is at least one such vertex for this case.

A local set cannot contain less than two vertices by its properties, and in this case there are two legs with one vertex of $L$ each,  one standard leg without any vertices of $L$, and $v \notin L$. Every vertex of $L$ separates every pair $x,y \notin L$, as $|L|=2$.
For every leg $\ell$, if $\ell^i \in L$ for $i \geq 1$ then $\ell$ has $i$ vertices, as otherwise $\ell^i$ does not separate $\ell^{i-1}$ and $\ell^{i+1}$ if $i\geq 2$, and it does not separate $v$ and $\ell^2$ if $i=1$. As there is a leg with a type $(s,0)$ solution, the solution of $\ell$ is of type $(s,2)$ or $(s,3)$, and since it has one vertex of $L$, the solution must be of type $(s,3)$. Thus, $i=1$, and $\ell$ is short. On the other hand, if $L$ consists of two vertices of short legs, the remaining vertices are $v$ and the vertices of one leg $\tilde{\ell}$, each having a different distance to the two vertices of $L$ (this distance is $1$ for $v$, and $i+1$ for $\tilde{\ell}^i$).
\end{proof}

Consider a tree with two small cores $v$ and $u$, and no regular cores. The tree consists of a path between $u$ and $v$, and each of them has two standard legs, a short leg and another standard leg which can be short or long. For each small core, the path to the other small core and its legs can be seen as its modified leg.

\begin{lemma}
Consider a tree with no regular cores and two small cores $v$ and $u$. A set $L \subseteq V$ is a landmark set if and only if it contains a local set for the three g-legs of $v$ and it is a local set for the three g-legs of $u$.

A minimal landmark set, with respect to set inclusion, will contain at most one vertex on the path between $u$ and $v$ (excluding $u$ and $v$), and at most two vertices of each long leg of $u$ and $v$.
\end{lemma}
\begin{proof}
Since the tree can be seen as a core and its three g-legs in two way, $L$ must contain a local set for the g-legs of $u$ and for the g-legs of $v$.

Assume now that $L$ contains local sets for the two sets of g-legs. If for at least one small core, each of its g-legs has a vertex of $L$, we are done by the first part of Lemma \ref{oneforeach}. Moreover, if $u\in L$ or $v \in L$, we are done by the fourth part of Lemma \ref{oneforeach}.
Since a modified leg has at least one vertex of a local set (by property $2$ of local sets) we find that each of $u$ and $v$ has a standard leg without any vertex of $L$, and every modified leg has at least two vertices of $L$.

If $L$ contains a vertex $z$ on the path between $u$ and $v$ (excluding $u$ and $v$), for the sake of the proof we see $z$ as having two g-legs (one containing $u$ as a small core, and the other one containing $v$ as a small core). As each of $u$ and $v$ has a standard leg with at least one vertex in $L$, each g-leg of $z$ has at least one vertex of $L$. Consider two vertices $x,y \notin L$. As $z \in L$, each of $x$ and $y$ is on a g-leg of $z$. If $x$ and $y$ are on the same g-leg, assume (without loss of generality) that this is the g-leg containing $u$. As the solution of one of the standard legs of $u$ is of type $(s,0)$, its other leg has either a type $(s,2)$ solution or a type $(s,3)$ solution. If $d(x,z)=d(y,z)$, then $d(x,u)=d(y,u)$ (as the paths of $x$ and $y$ to $z$ traverse $u$), $x$ and $y$ must be the vertices in position $1$ on the standard legs of $u$, and in this case the solution type of the standard leg of $u$ with at least one vertex of $L$ cannot be $(s,3)$ (as $x,y \notin L$). The two vertices of $L$ of the long standard leg of $u$ are closer to the vertex of position $1$ of that leg than to the vertex of position $1$ of the short leg of $u$, and thus they separate $u$ and $v$. If $d(x,z) \neq d(y,z)$, then one of $x$ and $y$ is closer to $z$ than the other vertex, and it is also closer to any vertex of $L$ on the other g-leg of $z$. This results in at least two separations between $x$ and $y$. If $x$ is on the g-leg of $z$ containing $u$ while $y$ is on the g-leg of $z$ containing $v$, let $u' \in L$ be on the former g-leg of $z$ and let $v' \in L$ be on the latter g-leg of $z$. We consider the two cases again. Assume that $d(x,z) \leq d(y,z)$. We get $d(x,u')\leq d(x,u)+d(u,u')$ while $d(y,u')=d(x,u)+d(u,u')$, and $d(x,u)<d(x,z) \leq d(y,z)$, showing that $u'$ separates $x$ and $y$. If $d(x,z) < d(y,z)$, $z$ also separates $x$ and $y$, and if $d(x,z) = d(y,z)$, then $d(y,v')\leq d(y,v)+d(v,v')$ while $d(x,v')=d(x,v)+d(v,v')$, and $d(y,v)<d(y,z)=d(x,z)<d(x,v)$, showing that $v'$ also separates $x$ and $y$.

We are left with the case that each modified leg has at least two vertices, and these vertices are not on the path between $u$ and $v$ (and they are not $u$ or $v$), and as one standard leg of each small core has no vertices of $L$, each small core has one long leg with two vertices of $L$. In this case consider two vertices $x,y \notin L$. Each core has two g-legs with at least two vertices of $L$ on each. If $d(x,u)=d(y,u)$, then there are two separations between $x$ and $y$ as $L$ contains a local set for $u$ (by Lemma \ref{lemma3}). Otherwise, if $d(x,u)<d(y,u)$, there is a g-leg of $u$ that is not the g-leg of $y$ and has two vertices of $L$, and these vertices separate $x$ and $y$.

Next, consider a minimal landmark set $L$. If $L$ has at least two vertices on the path between $u$ and $v$ excluding the endpoints, then the modified legs of $u$ and $v$ (both containing this path) have at least three vertices of $L$ each (as each of $u$ and $v$ has at least one vertex of $L$ on a standard leg). Removing one vertex of the path between $u$ and $v$ does not change the type of solutions of modified legs (a solution of type $(m,2)$ remains of type $(m,2)$ and a solution of type $(m,3)$ remains of type $(m,3)$). Next, assume that a leg of one small core has at least three vertices. Removing the vertex of maximum distance to the small core of its standard leg does not change the type of solution of this leg (the solution was of type $(s,2)$ and remains of this type).
\end{proof}

We say that a local set is thrifty if solutions of types $(s,2)$, $(m,2)$ and $(m,3)$ consist of exactly two vertices.

\begin{lemma}
\begin{enumerate}

\item Consider a landmark set $L$, such that the subset of vertices of $L$ on the g-legs of a regular core $v$ is not a thrifty local set. Then, $L$ is not minimal with respect to set inclusion.

\item Consider a landmark set $L'$ for a tree with a single core $u$ that is small, such that the subset of vertices of $L'$ on the standard legs of $u$ is not a thrifty local set. Then, $L'$ is not minimal with respect to set inclusion.
\end{enumerate}
\end{lemma}
\begin{proof}
Let $S'=L'\setminus \{u\}$. If $S'$ is not thrifty, we find $|S'|\geq 4$, as there is a leg $\ell'$ with a type $(s,2)$ solution consisting of at least three vertices, and at least one of the other two legs has at least one vertex in $S'$. Removing an arbitrary vertex of $\ell'$ from $S'$ to obtain $S''$ still results in a local set with at least three vertices,  and by Lemma \ref{threev}, $S''$ is a landmark set for the tree, showing that $L'$ is not minimal.

Let $S$ be the subset of $L$ restricted to the g-legs of $v$. By Lemma \ref{mustbelocal}, $S$ is a local set. Since $S$ is not thrifty, there is a g-leg $\ell$ that contains at least three vertices. If $\ell$ has a type $(s,2)$ solution, remove an arbitrary vertex of $\ell$ from $L$ to obtain $\tilde{L}$. If $\ell$ has a type $(m,2)$ or type $(m,3)$ solution, remove a vertex of $\ell$ from $L$ to obtain $\tilde{L}$, where the removed vertex is such that the solution remains of the same type (if the small core has position $i$ on $\ell$, then for a type $(m,2)$ solution, remove a vertex such that at least two vertices of positions $i+2$ or larger on $\ell$ are not removed. For a type $(m,3)$ solution, remove a vertex whose position on $\ell$ is not $i+1$).
In all cases,  the modification does not change the local sets of other regular cores, and the subset of $\tilde{L}$ restricted to the g-legs of $v$ remains a local set, thus by Lemma \ref{localsetsaregreat}, $\tilde{L}$ is a landmark set for the tree, showing that $L$ is not minimal.
\end{proof}

The case of a tree with two cores that are small was not considered in the lemma as in this case there can be a minimal landmark set $L$ that is not thrifty.

\section{The algorithm} The algorithm consists of two parts. In the first part, we detect the cores, and partition them into regular cores and small cores. Every regular core will have a list of its neighbors whose subtrees are g-legs. Finding such a list can be done by running DFS on the tree. In the second part (which we describe in more detail below), if the tree has at least one regular core, the algorithm simply finds a thrifty local set for each regular core. If the tree does not have any regular cores, then a landmark set is computed by considering all possible solution types. The resulting running times are linear.

In the case of a tree with at least one regular core $u$, the local sets are computed independently, and moreover, the dependence between the g-legs is only in the sense that for a given g-leg of $u$, only the types of solutions of the other g-legs are relevant, and not the specific solutions of the other g-legs of $u$. Thus, for each g-leg, we search for a solution of minimum cost for a given solution type. The case of a single small core is similar, as the identity of the vertex in a type $(s,1)$ solution or the vertices in a type $(s,2)$ solution does not affect the validity of the local set as a landmark set (this property only depends on number of vertices). In the case of two small cores, there are several similar properties. If there is a vertex on the path between the two cores (excluding the cores) in a landmark set, its exact identity is not important since replacing it with another vertex of this path keeps the types of solutions of the modified legs as they were. Similarly, for a long leg of one of the small cores, replacing one vertex whose position is at least $2$ with another such vertex does not change the solution type (neither for the long leg nor for the modified leg of the other small core that contains this long leg).

\paragraph{Finding a minimum cost local set for a regular core.}
Consider a regular core $v$. At most three solutions (which are thrifty local sets for the g-legs of $v$) will be considered, and a solution of minimum cost will be given as the output. The solution kinds are based on the definitions of local sets and thrifty local sets.

First, the different kinds of solutions for each g-leg are computed. For every standard leg compute the minimum cost solutions of types $(s,1)$, $(s,2)$, and $(s,3)$ (by finding the minimum cost vertex whose position is not $1$, and the two minimum cost vertices, using linear time in the length of the leg). For short legs there a unique solution is computed, which has type $(s,3)$.
For every modified leg $l'$ compute the minimum cost solutions of types $(m,1)$, $(m,2)$, and $(m,3)$ (if the small core has position $i$ on $l'$, find the minimum cost vertex $b$ out of the two vertices of positions $i+1$ on $l'$ , the two minimum cost vertices of positions $i+2$ or greater on $l'$, and a vertex of minimum cost excluding $b$ (to be combined with $b$ in an output set), using linear time in the number of vertices of the leg).

The first solution is computed independently for each g-leg. For every standard leg, select a solution of minimum cost, out of the solutions computed for it. In the second solution, there will be a short leg whose solution is of type $(s,0)$ (if there is no such leg of $v$, then no such solution is computed). For each modified leg, select a minimum cost solution out of its already calculated solutions of types $(m,2)$ and $(m,3)$. For each long leg, select a minimum cost solution out of its already calculated solutions of types $(s,2)$ and $(s,3)$. Select a short leg whose $(s,3)$ type solution has maximum cost and change its solution into type $(s,0)$. This completes the description of the second solution.
In the third solution, there will be a long leg of $v$ whose solution is of type $(s,0)$ (if there is no such leg, then no such solution is computed). For each modified leg, select a minimum cost solution out of its already calculated solutions of types $(m,2)$ and $(m,3)$. For each short leg, the solution is of type $(s,3)$. Find a long leg whose $(s,2)$ solution has the maximum cost. Define the solution of this leg to be of type $(s,0)$, and any other long leg will have a type $(s,2)$ solution.

\paragraph{Finding a minimum cost landmark set for a tree with no regular cores and one small core.} Let $v$ denote the small core of the tree.
In this case we will consider solutions of several kinds, and as not every local set with two vertices is a valid landmark set, we will consider local sets of two vertices separately. As in the case of a regular core, the first step is to computed minimum cost solutions of the three types ($(s,1)$, ($s,2)$, and $(s,3)$) for each leg.

The first solution is computed as for regular cores, that is, a minimum cost solution out of the three types is selected for each leg. Next, for each of the three legs, local sets where this leg has a type $(s,0)$ solution are considered. Consider a leg $\ell$ whose solution will be of type $(s,0)$. If $\ell$ is long, a type $(s,2)$ solution is selected for any long leg except for $\ell$, and a type $(s,3)$ solution is selected for any short leg. This is a landmark set as it either contains at least three vertices, or $L$ consists of two vertices the two of short legs, by Lemma \ref{threev}. If $\ell$ is short, then there are (at most) four additional solutions to be considered for the case where $\ell$ has a type $(s,0)$ solution, where any short leg except for $\ell$ has a type $(s,3)$ solution, and any long leg has either a type $(s,2)$ solution or a type $(s,3)$ solution. The only kind of solution with two vertices of the legs is the one where two short legs have type $(s,3)$ solution, resulting in two selected vertices. In this case, $v$ is added to the solution if there is a long leg with a type $(s,3)$ solution. We found at most four solutions for each leg, giving a constant number of solutions, where the output is a minimum cost solution out of these solutions.

\paragraph{Finding a minimum cost landmark set for a tree with no regular cores and two small cores.}
Recall that in this case the output may contain a local set that is not thrifty. Let the two cores be denoted by $u$ and $v$. Let $a_{uv}$ denote a vertex of minimum cost on the path between $u$ and $v$ excluding the endpoints. Let $u_1$, $u_2$, $v_1$, $v_2$ denote the four neighbors of $u$ and $v$ on theirs standard legs (the vertices of position $1$). Let $b_1$ and $b_2$ be two vertices of minimum costs on the long leg of $v$ that are not neighbors of $v$, and let $b'_1$ and $b'_2$ be two vertices of minimum cost on the long leg of $u$ that are not neighbors of $u$ (it is possible that some of the vertices $b_1$, $b_2$, $b'_1$, and $b'_2$ do not exist if at least one of $u$ and $v$ does not have a long leg, or it has a long leg with two vertices). Consider all subsets of $\{a_{uv},u,v,u_1,u_2,v_1,v_2,b_1,b_2,b'_1,b'_2\}$ .
For each subset, test whether it is a local set for $u$ and a local set for $v$ (or alternatively, test whether it is a landmark set, which can be done in linear time by computing the distances from each of the eleven vertices to all vertices), and let the sets satisfying this property be called valid solutions. As the set $\{u_1,u_2,v_1,v_2\}$ contains a local set for $u$ and a local set for $v$, at least one valid solution is found. Return a subset of minimum cost out of the valid solutions.

\end{document}